\newcommand{\extended}[1]{}         

\documentclass{article}
\usepackage{kr}

\usepackage{times}
\usepackage{soul}
\usepackage{url}
\usepackage[hidelinks]{hyperref}
\usepackage[utf8]{inputenc}
\usepackage[small]{caption}
\usepackage{graphicx}
\usepackage{amsmath}
\usepackage{amsthm}
\usepackage{booktabs}
\usepackage{algorithm}
\usepackage{algorithmic}
\usepackage[usenames,dvipsnames]{xcolor}

\usepackage{natbib}
\usepackage{amsthm}
\usepackage{amsfonts,amsmath,latexsym,wasysym,dsfont,stmaryrd}
\usepackage{bm,xspace}
\usepackage{url} 
\usepackage{listings}

\theoremstyle{definition}
\newtheorem{definition}{Definition}

\theoremstyle{plain}

\theoremstyle{remark}

\newcommand{\halfline}{\vspace{.5\baselineskip}}
\newcommand{\head}[1]{\noindent \textbf{#1}}

\renewcommand{\dots}{...}

\newcommand{\ie}{i.e.}
\newcommand{\eg}{e.g.}

\newcommand{\colonequals}{:=}

\newcommand{\last}{\mbox{last}}

\newcommand{\APf}{\textnormal{AP}}
\renewcommand{\phi}{\varphi}

\newcommand{\prop}[1]{\ensuremath{\mathsf{{#1}}}}

\newcommand{\until}{{\bf U}}
\newcommand{\Until}{\until}
\newcommand{\U}{\until}

\newcommand{\X}{{\bf X}}
\newcommand{\G}{{\bf G}}
\newcommand{\F}{{\bf F}}

\newcommand{\profile}[1]{\boldsymbol{#1}}

\newcommand{\Ag}{\textnormal{Ag}}
\newcommand{\ag}{a}
\newcommand{\agb}{b}
\newcommand{\Act}{\textnormal{Ac}}

\newcommand{\act}{c}
\newcommand{\mov}{\profile{\act}}
\newcommand{\mova}[1][\ag]{\mov_{#1}}

\newcommand{\setpos}{St}
\newcommand{\States}{\setpos}
\newcommand{\pos}{s}
\newcommand{\state}{\pos}
\newcommand{\CGS}{\relax\ifmmode\mathcal G\else$\textrm{CGS}$\xspace\fi}
\newcommand{\System}{\mathcal G}
\newcommand{\trans}{\delta}
\newcommand{\legal}{\textnormal{L}}
\newcommand{\val}{\ell}

\newcommand{\iplay}{\pi}
\newcommand{\strat}{\sigma}

\newcommand{\setstrat}{\mbox{\emph{Str}}}
\newcommand{\setstrata}[1][\ag]{\mbox{\emph{Str}}^r_{#1}}

\newcommand{\obsrel}[1][\ag]{\sim_{#1}}

\newcommand{\history}{h}
\newcommand{\History}{\text{Hist}}

\newcommand{\Reals}{\mathbb{R}}

\newcommand{\wCGS}{\relax\ifmmode \mathcal G\else$\textrm{wCGS}$\xspace\fi}

\newcommand{\lan}[1]{\ensuremath{\mathsf{#1}}\xspace}

\newcommand{\ATL}[1][]{\lan{ATL_{\stratstyle{#1}}}}
\newcommand{\ATLs}[1][]{\lan{ATL_\stratstyle{#1}^*}}
\newcommand{\PATL}[1][]{\lan{PATL_{\stratstyle{#1}}}}
\newcommand{\PATLs}[1][]{\lan{PATL_\stratstyle{#1}^*}}
\newcommand{\SL}[1][]{\lan{SL_{\stratstyle{#1}}}}
\newcommand{\PSL}[1][]{\lan{PSL_{\stratstyle{#1}}}}

\newcommand{\stratstyle}[1]{\ensuremath{\mathrm{#1}}}
\newcommand{\ir}{\stratstyle{ir}\xspace}

\newcommand{\Dist}{\text{Dist}}

\newcommand{\distribution}{\mathsf{d}}

\newcommand{\coop}[1]{\langle\!\langle{#1}\rangle\!\rangle}
\newcommand{\noavoid}[1]{[\![{#1}]\!]}
\newcommand{\coalition}{C}
\newcommand{\agent}{a}

\newcommand{\complexityclass}[1]{\ensuremath{\mathbf{{#1}}}\xspace}

\newcommand{\Ptime}{\complexityclass{P}}
\newcommand{\PTIME}{\complexityclass{P}}
\newcommand{\NP}{\complexityclass{NP}}
\newcommand{\coNP}{\complexityclass{co}\text{-}\complexityclass{NP}}

\newcommand{\Pspace}{\complexityclass{PSPACE}}
\newcommand{\NPspace}{\complexityclass{NPSPACE}}
\newcommand{\PSPACE}{\Pspace}
\newcommand{\NPSPACE}{\NPspace}
\newcommand{\Exptime}{\complexityclass{EXPTIME}}
\newcommand{\EXPTIME}{\Exptime}

\newcommand{\Deltacomplx}[1]{\complexityclass{\Delta_{{#1}}^{\Ptime}}}

\newcommand{\Deltwo}{\Deltacomplx{2}}

\newif\ifdraft\drafttrue 

\ifdraft
\newcommand{\munyque}[1]{{\small\color{RoyalBlue}[MM:  #1]}}

\newcommand{\wojtek}[1]{{\small\color{ForestGreen}[WJ: #1]}}

\newcommand{\nello}[1]{{\small\color{Orange}[NM:  #1]}}

\else
\newcommand{\munyque}[1]{}
\newcommand{\nello}[1]{}
\newcommand{\wojtek}[1]{}

\fi

\usepackage{tikz}
\usetikzlibrary{arrows,positioning}
\usepackage{natbib}
\usepackage{multirow}

\newcommand{\WJ}[1]{{\small\color{ForestGreen}[WJ: #1]}}

\newtheorem{theorem}{Theorem}

\title{Strategic Abilities of Forgetful Agents in Stochastic Environments}

\author{Francesco Belardinelli$^{1}$\and
Wojciech Jamroga$^{2,3}$\and
Munyque Mittelmann$^4$\and
Aniello Murano$^4$ \\
\affiliations
$^1$Imperial College, London, UK\\
$^2$Interdisc. Centre on Security, Reliability and Trust, SnT, University of Luxembourg\\
$^3$Institute of Computer Science, Polish Academy of Sciences, Warsaw, Poland\\
$^4$University of Naples ``Federico II''\\
\emails
francesco.belardinelli@imperial.ac.uk,
wojciech.jamroga@uni.lu,
\emails\{munyque.mittelmann, aniello.murano\}@unina.it
}

\begin{document}

\maketitle

\begin{abstract}
In this paper, we investigate the probabilistic variants 
of the strategy logics \ATL and \ATLs under imperfect information. Specifically, we present novel 
decidability and complexity results when the model transitions are stochastic and agents play uniform strategies. 
That is, the semantics 
of the logics are based on multi-agent, stochastic transition systems 
with imperfect information, which combine two sources of uncertainty, 
namely, the partial observability agents have on the environment, and 
the likelihood of transitions to occur from a system state. Since the 
model checking problem is undecidable in general in this setting, we 
restrict our attention to agents with memoryless (positional) 
strategies. The resulting setting captures the situation in which agents 
have qualitative uncertainty of the local state and quantitative 
uncertainty about the occurrence of future events. We illustrate the 
usefulness of this setting with meaningful examples.
\end{abstract}

\section{Introduction}

Complex and interacting Multi-Agent Systems (MAS) often face different kinds of uncertainty. One of the sources of uncertainty is the inability to completely observe the current local situation  (\eg, whether there is public transport available to the target destination). On the other hand, the occurrence of many natural events and the future behaviour of other agents, while it cannot be known with certainty, can be measured based on experiments or past observations. For instance, while we cannot know if the bus is going to arrive on time, we may have observed that this happens $0.7\%$ of the time.  Clearly, intelligent autonomous agents need to consider both the imperfect information about the local state and the likelihood of stochastic events when making strategic decisions and plans.

To see this, consider, for instance, the problem of online mechanisms, 
  which are  preference aggregation games in dynamic environments with multiple
agents and private information. 
Many multi-agent problems are inherently dynamic 
rather than static. Practical  examples include the problem of allocating computational resources (bandwidth, CPU, etc.) to processes arriving over time, selling items to a possibly changing group of buyers with uncertainty about the future supply, and selecting employees from a dynamically changing list of candidates \citep{Nisan2007}. 

\emph{Probabilistic model-checking}
is a technique for the formal and automated analysis of probabilistic systems that can be modeled by stochastic state-transition models \citep{clarke2018handbook}. Its  
aim is to establish the correctness of such systems  against  probabilistic specifications, which may describe, \eg, the probability of an unsafe event to occur, or the ability of a coalition to ensure the completeness of a task. 

Logic-based approaches have been widely and successfully applied for probabilistic verification of MAS. For instance, probabilistic model-checking techniques have been used for verification of preference aggregation mechanisms \citep{mittelmann2023bayesian}, negotiation games \citep{Ballarini2009}, team formation protocols \citep{chen2011verifying}, and stochastic behaviors in dispersion games \citep{hao2012probabilistic}, to name a few. 
\cite{KR2021-30} investigates the problem of deciding whether the probability of satisfying a given temporal formula in a concurrent stochastic game is 1 or greater than 0. 
 \cite{Kwiatkowska2022} details how verification techniques can be developed and implemented for concurrent stochastic games. 

In this paper, we consider logics for reasoning about strategic abilities  while  taking into account both incomplete information and probabilistic behaviors of the environment and agents. 
We study  the  Probabilistic Alternating-time Temporal Logics \PATL and \PATLs \citep{chen2007probabilistic,hao2012probabilistic} under imperfect information (II) for a classic type of agents \citep{fagin2004reasoning} called  imperfect-recall (that is, agents who use memoryless strategies, also called Markovian strategies or policies).
Model checking \PATLs under II for agents with perfect-recall (who uses memoryful strategies) is known to be undecidable in general  even for the fragment with a single-player \citep{hao2012probabilistic}.  
We introduce and motivate the problem of strategic reasoning under combined types of uncertainty and memoryless agents. We then provide results on the model-checking complexity for \PATL with memoryless deterministic strategies for the coalition and point directions to challenging open questions. 

\smallskip
\noindent\textbf{{Related Work.}}
Recently, much work has been done on logics for strategic reasoning in Multi-Agent Systems, starting from the pioneering work on Alternating-time Temporal Logics
\ATL and \ATLs \citep{AlurHK02}.
These logics enable reasoning about the strategic abilities of agents in a cooperative or competitive system.   
\ATL\ has been extended in various directions, considering for instance strategy contexts~\citep{DBLP:journals/iandc/LaroussinieM15} or adding imperfect information ~\citep{jamroga2011comparing}. Strategy Logic (\SL) \citep{ChatterjeeHP10,MMPV14} extends \ATL to treat  strategies as first-order variables. 

Contexts of imperfect information have been extensively considered in the literature on formal verification (see, for instance, \citep{Dima11undecidable,KV00,JA06,Reif84,BJ14,tocl-BMMRV21,journals-BLMR20,BD08}). Generally, imperfect information in MAS entails higher complexity, which may be even undecidable when considered in the context of memoryful strategies~\citep{Dima11undecidable}. In order to retrieve a decidable model-checking problem, it is interesting to study imperfect information MAS with memoryless agents~\citep{journals-CLMM18}. 

Several works consider the verification of systems against specifications given in probabilistic logics. In particular, 
\cite{wan2013model} study the model-checking problem for Probabilistic Epistemic Computational Tree Logic with semantics based on  probabilistic interpreted systems. 
In the context of MAS,  
\citep{huang2013logic} studies an ATL-like logic for stochastic MAS in a setting in which agents play deterministic strategies and have probabilistic knowledge about the system. 
 \citep{fu2018model} shows model-checking an epistemic logic with temporal operators under strategies that depend only on agents' observation history is undecidable.

\cite{chen2007probabilistic} 
propose model-checking algorithms for Probabilistic ATL in the perfect information setting. Perfect information was also considered with specification in Probabilistic Alternating-Time $\mu$-Calculus \citep{song2019probabilistic} and Probabilistic Strategy Logic \cite{aminof2019probabilistic}. 
ATL-based probabilistic logics were also considered for the verification of unbounded parameterized MAS \citep{lomuscio2020parameterised},  for resource-bounded MAS \citep{nguyen2019probabilistic}, and under assumptions over opponents' strategies \citep{DBLP:journals/fuin/BullingJ09}. 

The closest related work is \citep{Huang2012}, which considers the logic \PATLs under incomplete information and synchronous perfect recall. The complexity results show that the model-checking problem is in general undecidable even for the single-agent fragment of the logic. 

Also related are the works in \citep{gripon2009qualitative,doyen2011games,carayol2018pure,Doyen2022}, which consider algorithmic solutions for computing the existence of winning strategies and winning distributions for two-player stochastic games with imperfect information. Finally, \cite{gurov2022knowledge} investigate the problem of strategy synthesis for knowledge-based strategies against a non-deterministic environment.

 \section{Preliminaries}
\label{sec:preliminars}

In this paper, we fix finite non-empty sets of agents $\Ag$, actions $\Act$,
atomic propositions $\APf$. 
We  write $\profile{o}$ for a tuple of objects $(o_\ag)_{\ag\in\Ag}$, one for each agent, and such tuples are called \emph{profiles}. A \emph{joint action} or \emph{move} $\mov$ is an element of $\Act^{\Ag}$. 
Given a profile $\profile{o}$ and $\coalition\subseteq\Ag$, we let $o_\coalition$ be the components of agents in  $\coalition$, and $\profile{o}_{-\coalition}$ is $(o_\agb)_{\agb\not \in \coalition}$. Similarly, we let $\Ag_{-\coalition}=\Ag\setminus\coalition$. 

\halfline
\head{Distributions. } Let $X$ be a finite non-empty set. A \emph{(probability) distribution} over $X$ is a function $\distribution:X \to [0,1]$ such that $\sum_{x \in X} \distribution(x) = 1$, and $\Dist(X)$ is the set of distributions over $X$. We write $x \in \distribution$ for $\distribution(x) > 0$. 
If $\distribution(x) = 1$ for some element $x \in X$, then $\distribution$ is a \emph{point (a.k.a. Dirac) distribution}. 
If, for $i\in I$, $\distribution_i$ is a distribution over $X_i$, then, writing $X = \prod_{i\in I} X_i$, the \emph{product distribution} of the $\distribution_i$ is the distribution $\distribution:X \to [0,1]$ defined by $\distribution(x) = \prod_{i\in I} \distribution_i(x_i)$.

\halfline
\head{Markov Chains. }
  A \emph{Markov chain} $M$ is a tuple $(\setpos,p)$ where $\setpos$ is a set of states and $p \in \Dist(\setpos \times \setpos)$ is a distribution. The values $p(s,t)$ are called \emph{transition probabilities} of $M$.
  A \emph{path} is an infinite sequence of states.

\halfline
\head{Concurrent Game Structures. } 
A \emph{stochastic concurrent game structure with imperfect information} (or simply \emph{CGS})
  $\System$ is a tuple $(\setpos, \legal, \trans, \val, \{\obsrel\}_{\ag\in\Ag})$ where
(i) $\setpos$ is a finite non-empty set of \emph{states};
(ii) $\legal: \setpos \times \Ag \to 2^\Act\setminus\{\emptyset\}$ is a \emph{legality function} defining the available actions for each agent in each state, we write $\profile{\legal(\pos)}$ for the tuple $(\legal(\pos, \ag))_{\ag\in\Ag}$; (iii) for each state $\pos \in \setpos$ and each  move $\mov \in \profile{\legal(\pos)}$, the \emph{stochastic transition function} $\trans$ gives the (conditional) probability $\trans(\pos, \mov)$ of a transition from state $\pos$ for all $\pos' \in \setpos$ if each player $\ag \in \Ag$ plays the action $\mova$, we also write this probability as $\trans(\pos, \mov)(\pos')$, to emphasize that $\trans(\pos, \mov)$ is a probability distribution on $\setpos$;
(iv) $\val:\setpos \to 2^{\APf}$ is a \emph{labelling function};
(v)  
    $\obsrel\;\subseteq \setpos\times\setpos$ is an equivalence relation called the  \emph{observation relation} of agent $\ag$.

Throughout this paper, we assume that the CGS is uniform, that is, if two states are indistinguishable for an agent $\ag$, then $\ag$ has the same available actions in both states. Formally, if $\pos \obsrel \pos'$ then $\legal(\pos, \ag) = \legal(\pos', \ag)$, for any $\pos, \pos' \in \setpos$ and $\ag \in \Ag$. For each state $\pos \in \setpos$ and joint action $\mov \in \prod_{\ag \in \Ag} \legal(\pos,\ag)$, we also assume that there is a state $\pos'\in\setpos$ such that $\trans(\pos, \mov)(\pos')$ is non-zero, that is, every state has a successive state from a legal move.

We say that $\System$ is \emph{deterministic} (instead of stochastic) if every $\trans(\pos,\mov)$ is a point distribution. 

\halfline
\head{Plays. } 
A \emph{play} or path in a CGS $\System$ is an infinite sequence $\iplay=\pos_0 \pos_1 \cdots$ of states
such that there exists a sequence $\mov_0 \mov_1 \cdots$ of joint-actions such that $\mov_i \in \legal(\pos_{i})$ and  $\pos_{i+1} \in \trans(\pos_i,\mov_i)$ (\ie, $\trans(\pos_i,\mov_i)(\pos_{i+1} )>0$) for every $i \geq 0$.
We write $\iplay_i$ for $\pos_i$,
$\iplay_{\geq i}$ for the suffix of
$\iplay$ starting at position $i$. 
Finite paths are called \emph{histories}, and the set of all histories is denoted $\History$. Write $\last(\history)$ for the last state of a history $\history$.

\halfline
\head{Strategies. }
A (general) \emph{probabilistic strategy} is a  function $\strat:\History
\to  \Dist(\Act)$ that maps each history to a distribution of 
actions.  
We let $\setstrat$ be the set of all strategies. A \emph{memoryless uniform  probabilistic strategy} for an agent $\ag$
is a function $\sigma_\ag: \setpos \to \Dist(\Act)$ 
in which for all positions $\pos,\pos'$ such that $\pos\obsrel\pos'$, we have $\strat(\pos)=\strat(\pos')$. We let $\setstrata$ be the set of uniform strategies for agent $\ag$.
A deterministic (or \emph{pure}) strategy $\strat$ is a strategy in which  $\strat(\pos)$ is a point distribution for any $\pos$. 
A \emph{strategy profile} is a tuple $\profile\strat$  of strategies, one for each agent. We write $\strat_\ag$ for the strategy of  $\ag$ in the strategy profile $\profile\strat$. 
For a  strategy $\strat_\ag$ for  agent $\ag$, we assume that $\profile\strat(\history)(\act) = 0$ if $\act \not \in \legal(\last(\history),\ag)$. 
 \section{Probabilistic \ATL and \ATLs}
We begin by introducing the Probabilistic Alternating-Time Temporal Logics \PATLs and \PATL. 

The syntax of \PATLs  is defined by the grammar
\begin{align*}
	\varphi  ::= p \mid  {\varphi \lor  \varphi} \mid \neg \varphi \mid \X \varphi \mid \varphi \until \varphi \mid \coop{\coalition}^{\bowtie d} \varphi
\end{align*}
where $p \in \APf$, $\coalition \subseteq \Ag$, $d$ is a rational constant in
$[0, 1]$, and $\bowtie \in 
\{\leq, <, >, \geq\}$. 

The intuitive reading of the operators is as follows: $\coop{\coalition}^{\bowtie d}\varphi$ means that there exists a strategy for the coalition $\coalition$ to collaboratively enforce $\varphi$  with a probability in relation $\bowtie$ with constant $d$,
``next'' $\X$ and ``until'' $\U$ are the standard temporal operators.   
We  make use of the usual syntactic sugar ${\F \varphi \colonequals \top \U \varphi}$ and ${\G \varphi \colonequals \neg \F \neg \varphi}$ for temporal operators.
Finally, we use ${\noavoid{\coalition}^{\bowtie d}\varphi\colonequals \neg \coop{\coalition}^{\bowtie d} \neg \varphi}$ to express that no strategy of $\coalition$ can prevent $\varphi$  with a probability in relation $\bowtie$ with constant $d$.

An \PATLs formula of the form $\coop{\coalition}^{\bowtie d} \varphi$ or $\noavoid{\coalition}^{\bowtie d} \varphi$ is also called state formula. An important syntactic restriction of \PATLs, namely  \PATL, is defined as follows.

The syntax of \PATL  is defined by the grammar
	\begin{align*}   	
		\varphi ::= p \mid  \varphi \lor \varphi \mid \neg \varphi \mid \coop{\coalition}^{\bowtie d} \X \varphi \mid \coop{\coalition}^{\bowtie d}(\varphi \until\varphi)	
  \end{align*}
where $p \in \APf$, $\coalition \subseteq \Ag$, and $\bowtie \in \{\leq, <, >, \geq\}$.

Formulas of \PATL and \PATLs are interpreted over CGSs. 

 \halfline \head{Probability Space on Outcomes. } An \emph{outcome} of a
 strategy profile $\profile\strat$ and a state $\pos$ is a play $\iplay$ that starts with
 $\pos$ and is extended by $\profile\strat$, i.e., $\iplay_{0} =
 \pos$, and for every $k \geq 0$ there exists $\mov_k \in
 \profile\strat(\iplay_k)$ such that $\iplay_{k+1} \in
 \trans(\iplay_k,\mov_k)$. 
 The set of outcomes of a strategy profile  $\profile\strat$ and state $\pos$  is denoted $Out(\profile\strat,\pos)$.   
 A given system
 $\System$, strategy profile $\profile{\strat}$, and state
 $\pos$ induce an infinite-state Markov chain
 $M_{\profile{\strat},\pos}$ whose states are the finite prefixes of
 plays in
 $Out(\profile{\strat},\pos)$. Such finite prefixes of plays
   are called \emph{histories} and written
   $\history$, and we let
   $\last(\history)$ denote the last state in $\history$.  Transition probabilities
in  $M_{\profile{\strat},\pos}$ are defined as  $p(\history,\history\pos')=\sum_{\mov\in\Act^\Ag}
 \profile{\strat}(\history)(\mov) \times
 \trans(\last(\history),\mov)(\pos')$.  
 The Markov chain
 $M_{\profile{\strat},\pos}$ induces a canonical probability space on
 its set of infinite paths~\citep{kemeny1976stochastic}, which can be identified with the set of plays in  $Out(\profile{\strat},\pos)$ and the corresponding measure is denoted $out(\profile{\strat},\pos)$.
~\footnote{This is a classic construction, see for instance
 ~\citep{clarke2018model,berthon2020alternating}.}

  Given a coalition strategy $\profile{\strat_\coalition} \in \prod_{\ag \in \coalition} \setstrata$, we let $n = |\Ag\setminus\{\coalition\}|$ and define
the set of possible outcomes of $\profile{\strat_\coalition}$ from a state $\pos \in \setpos$ to be the set  $out_\coalition(\profile{\strat_\coalition},\pos) = \{out((\profile{\strat_\coalition},\profile{\strat_{-\coalition}}),\pos) : \profile{\strat_{-\coalition}} \in \setstrat^{n} \}$ of probability
measures that the players in $\coalition$ enforce when they
follow the strategy $\profile{\strat_\coalition}$, namely, for each $\ag \in \Ag$,
player $\ag$ follows strategy $\strat_\ag$. We use $\mu^{\profile{\strat_\coalition}}_\pos$ to range over $out_\coalition(\profile{\strat_\coalition},\pos)$.

\halfline
\subsubsection{\PATL and \PATLs Semantics}
\PATL and \PATLs formulas are interpreted in a transition system $\System$  
and a path  $\iplay$,
\begingroup
\allowdisplaybreaks
\begin{align*}
 \System,\iplay &\models p & \text{ iff } & p \in \val(\iplay_0)\\
 \System,\iplay &\models \neg \varphi & \text{ iff } & \System,\iplay \not \models \varphi \\
 \System,\iplay &\models \varphi_1 \lor \varphi_2 & \text{ iff }&  
 \System,\iplay \models \varphi_1  \text{ or } \System,\iplay \models \varphi_2
 \\
\System, \iplay&\models \coop{\coalition}^{\bowtie d} \varphi & \text{ iff }  & \exists \profile{\strat_{\coalition}} \in \prod_{\ag \in\coalition} \setstrata  \text{ such that }  \\ & & & \forall \mu^{\profile{\strat_\coalition}}_{\iplay_0} \in out_\coalition(\profile{\strat_{\coalition}},\iplay_0)
\text{, } 
\\ & & &
\mu^{\profile{\strat_\coalition}}_{\iplay_0}(\{\iplay' : \System,\iplay' \models \varphi\}) \bowtie d \\
 \System,\iplay &\models \X \varphi & \text{ iff } & \System,\iplay_{\geq 1} \models \varphi \\
\System, \iplay  & \models \psi_1 \until \psi_2 & \text{ iff } &  \exists k \geq 0 \text{ s.t. } \System,\iplay_{\geq k} \models \psi_2 \text{ and } 
\\ & & &
\forall j \in [i,k).\, \,  \System,\iplay_{\geq j}\models \psi_1
\end{align*} 
\endgroup

\section{Strategic Reasoning under Uncertainty}
\label{sec:examples}

Many real-life scenarios require agents to interact in partially observable environments with stochastic phenomena. 
A natural application of strategic reasoning over both of these sources of uncertainty is card games, as the distribution of cards is a stochastic event and the hand of each agent is kept secret from the other players. 

Let us see a more detailed example based on online mechanism design\footnote{Previous work \citep{SLKF_KR21,DBLP:conf/ijcai/MittelmannMMP22,mittelmann2023bayesian} have shown how to encode notions from Mechanism Design (e.g., strategyprofness) using logics for strategic reasoning.} and, in particular, elections. 
While the majority of elections have a static set of candidates which is known upfront, there are contexts where
candidates appear over time. A classic example is hiring a committee: the candidates that will appear the next day to pass an interview are unknown, 
and the voters must decide immediately
whether to hire one of the current candidates or not \citep{do2022online}.  

In  online approval-based election \citep{do2022online}, there is a non-empty set of candidates $C = \{1, \dots, m\}$ and the goal is to select $k\leq 1$ candidates for a committee. 
In each state, an unseen candidate $j$ is presented and the  agents vote on whether to include the current candidate in the
 committee or not. The election  continue until the committee is completed or all candidates have been rejected. 
 For a candidate $j$, we let the propositions  $rejected_j$, $selected_j$, $interview_j$, denote whether candidate $j$ was already rejected,   whether she was selected to the committee, and whether she is been currently interviewed, resp. For each agent $\ag$, 
$likes_{\ag,j}$
 denotes whether  $\ag$ is currently willing to approve the  candidate $j$.

Agents know their own preferences,  that is, the candidates they like but are uncertain about others' preferences. 
Voters can distinguish the candidate currently interviewed, but are unaware of the next candidate to be presented (\ie, whether $\X interview_j$ holds in any given state).

In each state $\pos$, agents can either accept or reject the current candidate (actions $y$ and $n$, resp.).  The probability of selecting candidate $j$ being selected is determined by the transition function $\trans(\pos, \mov)$, according to the actions in $\mov$. 
If all agents accept (similarly, reject) a candidate, the system transitions to a state in which the candidate is selected (resp. rejected) with a probability equal to one. If there is no consensus on whether to accept the candidate, the probability to transition to a state in which the candidate is selected is given by a rational constant $p_{j, \mov}\in (0,1)$. Similarly, the probability of moving to a state where she is rejected is $1-p_{j, \mov}$.

The \PATL formula 
$$rejected_j \to  \neg \coop{\coalition}^{\geq 1} \F selected_j$$

represents that the coalition $\coalition$ cannot select a candidate that was already rejected.

The \PATLs formula 
$$\coop{\coalition}^{\geq \frac{1}{2}} \bigwedge_{\ag \in \coalition} \bigvee_{j \in C} likes_{\ag, j} \land \F selected_j$$
represents that the coalition $\coalition$ can ensure, with probability greater or equal to $\frac{1}{2}$ to select in the future at least one candidate liked by each agent in $\ag$, while 
$$\coop{\coalition}^{\geq \frac{1}{2}} \bigwedge_{\ag \in \coalition} \bigwedge_{j \in C} likes_{\ag, j} \land \F selected_j$$
states that they can ensure, with probability greater or equal to $\frac{1}{2}$, all their liked candidates are eventually selected.

The formula 
$$interview_j \to \coop{\coalition}^{\leq \frac{1}{4}}  \X selected_j$$
says that the probability the coalition $\coalition$ ensures the currently interviewed candidate is selected in the next state is at most~$\frac{1}{4}$.

\section{Model Checking Complexity}\label{sec:mc}

In this section, we look at the complexity of model-checking for \PATL. 
In particular, we show that the problem for \emph{memoryless deterministic strategies of the coalition} against probabilistic play of the other agents and a stochastic environment is no more complex than in standard (non-probabilistic) case.
The settings introduced in this paper include both deterministic and probabilistic memoryless strategies for the coalition and deterministic and stochastic CGSs.
This gives 4 semantic variants in total, but the case of deterministic strategies and deterministic CGSs consists of the standard setting for \ATL, whose complexity results are well-established.

The main technical result of this paper is as follows.

\begin{theorem}\label{prop:mcheck-det-patl}
Model checking \PATL[\ir]\footnote{As usual in the verification process, we denote no recall with r and imperfect information with i. }  with deterministic strategies for the coalition is \Deltwo-complete.
\end{theorem}
\begin{proof}
The lower bound follows by a reduction of \ATL[\ir] model checking, which is \Deltwo-hard~\citep{Jamroga06atlir-eumas}.
Given are: a pointed CGS $(M,q)$ and a formula $\coop{\coalition}\varphi$ of \ATL[\ir]. Note that $M$ can be seen as stochastic CGS with only Dirac probability distributions for transitions.
Recall that, in finite games, the opponents always have a deterministic best-response strategy to any given strategy $\sigma_{\coalition}$.
Thus, $M,q \models_{\ATL[\ir]} \coop{\coalition}\varphi$ 
iff the agents in $\coalition$ have a uniform deterministic memoryless strategy to enforce $\varphi$ on all paths 
iff they have such a strategy against all the probabilistic responses from $\overline{\coalition}$.
Since the set of best responses includes deterministic strategies of $\overline{\coalition}$ played against deterministic strategy $\sigma_{\coalition}$ in the deterministic CGS $M$, this is equivalent to saying that $M,q \models_{\PATL[\ir]} \coop{\coalition}^{\ge 1}\varphi$, which completes the reduction.

For the upper bound, we apply a similar procedure to that of \ATL[\ir]~\citep{DBLP:journals/entcs/Schobbens04}. For formulas of type $\coop{\coalition}^{\bowtie \distribution} \varphi$ without nested strategic modalities, we guess a strategy $\sigma_{\coalition}$, prune the model accordingly, and merge the remaining agents ($\overline{\coalition}$) into a single opponent. This yields a single-agent Markov Decision Process with full observability.
Then, we check the Probabilistic Computation Tree Logic formula $A^{\bowtie \distribution} \varphi$, which can be done in time $\NP\cap\coNP$~\citep{chen2007probabilistic}.
\nocite{Hansson94logic}

For nested strategic modalities, we proceed recursively (bottom up), which runs in time $\PTIME^{\NP\cap\coNP}$ = \Deltwo.
\end{proof}

\extended{

\begin{theorem}
Model checking \PATLs[\ir] with deterministic strategies for the coalition is \PSPACE-complete.
\end{theorem}
\begin{proof}
The lower bound is immediate from the corresponding problem for ATL$^*$, which is also \PSPACE-complete.

As for the upper bound, we can apply the same procedure as for ATL$^*$ in the same setting: for formulas of type $\coop{\coalition}^{\bowtie \distribution} \varphi$, we guess a strategy and prune the model accordingly. Then, we check the PCTL$^*$ formula $A^{\bowtie \distribution} \varphi$.
This procedure gives an algorithm in \NPSPACE = \PSPACE.
\end{proof}

\section{Model Checking Probabilistic Play}

Model checking of \PATL[\ir] is related to synthesis of memoryless policies for POMDPs, which is known to be in \PSPACE, as well as \NP-hard and \emph{sum-of-square-roots}-hard~\citep{Vlassis12memorylessPOMDP}.
However, the two problems differ significantly in several ways, and cannot be easily reduced to one another:
\begin{itemize}
\item Policy synthesis for POMDPs addresses non-nested 1.5-player games with arbitrary rewards. It looks for single-agent strategies that maximize the agent's expected reward, averaged over all execution paths and future time points. Importantly, the reward decreases with each step by a given temporal discount that is strictly smaller than 1. No less importantly, the proponent is playing against a purely reactive stochastic environment.

\item \PATL[\ir] model checking admits nested strategic properties in games with arbitrarily many players. It seeks coalitional strategies that maximize the probability of enforcing a binary reachability/safety goal against all probabilistic behaviors of the opponents. No temporal discounting is considered.
\end{itemize}

Our proofs in the rest of this section have been inspired by the results and proofs of~\citep{Vlassis12memorylessPOMDP}.
Notice in particular that it is not possible to employ the technique that we used in Section~\ref{sec:mcheck-det} for deterministic strategies, i.e., calling an oracle that guesses the best memoryless strategy. This is because there are \emph{infinitely many} probabilistic memoryless strategies, and hence the oracle Turing machine would either have to run in unbounded time, or allow for infinite branching.
In fact, synthesis of optimal probabilistic strategies is a special case of \emph{jointly constrained bilinear optimization}, which is a notoriously hard problem~\citep{AlKhayyal90bilinear}.
Fortunately, our case can be reduced to deciding the second level in the \emph{hierarchical theory of the reals}~\citep{Schaefer22beyond-ExistThReals} which is an extension of the \emph{existential theory of the reals} problem~\citep{Canny88PSPACE}.

\subsection{Probabilistic Play: Upper Bounds}
\label{sec:mcheck-prob-upper}

We begin by showing that the model checking problem is decidable in \EXPTIME.
Moreover, for the special case of formulas that include only the grand coalition of agents, the problem is in \PSPACE, analogously to memoryless synthesis for POMDPs. We will use reductions to the following decision problems.

\begin{definition}[Existential theory of the reals (Th$\Reals_\exists$)]\label{def:exist-th-reals}
The problem decides the truth of a first-order formula
$\Phi \equiv \exists x_1 \dots \exists x_n\ P(x_1,\dots,x_n)$
where $x_i$ are interpreted over $\Reals$, and $P$ is a monotone Boolean function of atomic predicates of the form $f_i(x_1,\dots,x_n) \ge 0$ or $f_i(x_1,\dots,x_n) > 0$, with each $f_i$ being a polynomial with rational coefficients.
\end{definition}
\begin{theorem}[\cite{Canny88PSPACE}]\label{prop:exist-th-reals}
Th$\Reals_\exists$ is in \PSPACE.
\end{theorem}

\begin{definition}[First-order theory of the reals (Th$\Reals$)]\label{def:FO-th-reals}
Analogously to Definition~\ref{def:exist-th-reals}, only with an arbitrary sequence of quantifiers $Q_1\dots Q_n$ allowed at the beginning of $\Phi$.
\end{definition}
\begin{theorem}[\cite{Renegar92existThReals}]\label{prop:FO-th-reals}
There is an algorithm for Th$\Reals$ that requires
$(md)^{n\cdot 2^{O(\omega)}}$ operations
and $(md)^{O(n)}$ calls to an oracle computing $P$,
where $m$ is the number of atomic predicates in $\Phi$,
$d$ is the maximal degree of the polynomials,
$n$ is the number of quantifiers,
and $\omega-1$ the number of quantifier alternations in $\Phi$.

\WJ{In our case, we will have $m = n = O(|\States|), \omega = 2, d = 1$. This yields an algorithm with complexity $n^{n\cdot 2^{O(1)}} + n^{O(n)} = 2^{O(n\cdot \log n)}$. }
\end{theorem}

\begin{theorem}\label{prop:prob-exptime}
Model checking \PATL[\ir] with probabilistic strategies for the coalition is in \EXPTIME.
\end{theorem}
\begin{proof}
Proof idea, case $\System,\iplay \models \coop{\agent}^{\bowtie d}\prop{p_1}\Until\prop{p_2}$: \WJ{to be expanded}

\begin{enumerate}
\item Reconstruct the model as follows: (i) add a ``sink'' state $q_{sink}$ with a deterministic self-loop where \prop{p_2} is false; (ii) for all the states where \prop{p_2} is true or \prop{p_1} is false, remove all outgoing transitions and replace them with an automatic transition to $q_{sink}$. That is, we stop looking at the rest of the path whenever \prop{p_2} has been achieved (and thus $\prop{p_1}\Until\prop{p_2}$ already succeded) or \prop{p_1} has been invalidated (and thus $\prop{p_1}\Until\prop{p_2}$ already failed).

\item Formulate the set of constraints as equalities over the vector of (real-valued) rewards $r_\state$ for $\state\in\States$, expressing the expected probability of success from state $\state$, and probabilistic decisions $prob_{\state,\alpha}$ for $\state\in\States, \alpha\in\Act$, expressing the probability with which agent $\agent$ takes action $\alpha$ at state $\state$.

\item Decide whether there exists a valuation of the variables $r_\state$ and $prob_{\state,\alpha}$, $\state\in\States, \alpha\in\Act$ that satisfies the above constraints plus $r_{\iplay_0} \bowtie d$.

\item \WJ{this works when there are no opponents. What about the general case? We call an oracle?}
\end{enumerate}

For $\coop{\coalition}\prop{p_1}\Until\prop{p_2}$, proceed by calling oracles that guess the strategies of individual agents, one by one. Each oracle runs in \PSPACE, hence the overall procedure is in \PSPACE, too.

For nested strategic modalities, proceed recursively (bottom-up), obtaining $\Ptime^\PSPACE = \PSPACE$, qed.
\end{proof}

\begin{theorem}\label{prop:prob-pspace}
Model checking \PATL[\ir] with probabilistic strategies for the coalition is in \PSPACE.
\end{theorem}
\begin{proof}
\end{proof}

\subsection{Probabilistic Play: Lower Bounds}
\label{sec:mcheck-prob-lower}

} 

\section{Discussion}
This paper analyses the verification of the strategic abilities of autonomous agents in  MAS while accounting for both incomplete information and probabilistic behaviours of the environment and agents. 
The setting considered in this paper is significant as MAS are often set in partially observable environments, whose evolution might not be known with certainty, but can be measured based on experiments and past observations. 
Verification of strategic abilities in the general setting with perfect recall is known to be undecidable, but the restriction to memoryless strategies is meaningful. 
We provided complexity results for deterministic strategies for the proponent coalition and point out different settings that are currently challenging open questions, based on probabilistic strategies for the proponent coalition. 

For solving the model checking problem w.r.t.~probabilistic strategies for the proponent coalition, notice that it is not possible to exploit the technique used in Section~\ref{sec:mc} for deterministic strategies, i.e., calling an oracle that guesses the successful memoryless strategy. This is because there are \emph{infinitely many} probabilistic memoryless strategies, and hence the oracle Turing machine would either have to run in unbounded time, or allow for infinite branching.
In fact, the synthesis of optimal probabilistic strategies is a special case of \emph{jointly constrained bilinear optimization}, which is a notoriously hard problem~\citep{AlKhayyal90bilinear}.
Additionally, techniques employed for partially observable Markov decision processes (see for instance \citep{Vlassis12memorylessPOMDP}) can not be easily adapted as they refer to single-agent abilities.
Moreover, the work on Probabilistic Alternating $\mu$-calculus \citep{song2019probabilistic} seems unhelpful in our case. First, it is known that Probabilistic Alternating $\mu$-calculus and \PATL are incomparable \citep{bulling2011alternating,song2019probabilistic}. Second, the work \citep{song2019probabilistic} only considers perfect information strategies. 
Finally, using the work on \PSL\citep{aminof2019probabilistic} does not seem the right direction either. Indeed, it 
only 
considers perfect information strategies. Additionally, the model checking problem for \PSL is 3-EXPTIME-complete, while we expect a much lower complexity in our setting.   
\clearpage

\section*{Acknowledgments}
This research has been supported by the PRIN project RIPER (No. 20203FFYLK), the PNRR MUR project PE0000013-FAIR, the InDAM project ``Strategic Reasoning in Mechanism Design'', the EU ICT-48 2020 project TAILOR (No. 952215), the NCBR Poland/FNR Luxembourg projects STV (POLLUX-VII/1/2019 and C18/IS/12685695/IS/STV/Ryan), SpaceVote (POLLUX-XI/14/SpaceVote/2023 and C22/IS/17232062/SpaceVote) and PABLO (C21/IS/16326754/PABLO), as well as the EU H2020 Marie Sklodowska-Curie project with grant agreement No 101105549.

\bibliographystyle{kr}

\begin{thebibliography}{}

\bibitem[\protect\citeauthoryear{Al-Khayyal}{1990}]{AlKhayyal90bilinear}
Al-Khayyal, F.
\newblock 1990.
\newblock Jointly constrained bilinear programs and related problems: An
  overview.
\newblock {\em Computers \& Mathematics with Applications} 19(11):53--62.

\bibitem[\protect\citeauthoryear{Alur, Henzinger, and
  Kupferman}{2002}]{AlurHK02}
Alur, R.; Henzinger, T.; and Kupferman, O.
\newblock 2002.
\newblock Alternating-time temporal logic.
\newblock {\em J. {ACM}} 49(5):672--713.

\bibitem[\protect\citeauthoryear{Aminof \bgroup et al\mbox.\egroup
  }{2019}]{aminof2019probabilistic}
Aminof, B.; Kwiatkowska, M.; Maubert, B.; Murano, A.; and Rubin, S.
\newblock 2019.
\newblock Probabilistic strategy logic.
\newblock In {\em Proc. of IJCAI 2019},  32--38.
\newblock ijcai.org.

\bibitem[\protect\citeauthoryear{Ballarini, Fisher, and
  Wooldridge}{2009}]{Ballarini2009}
Ballarini, P.; Fisher, M.; and Wooldridge, M.
\newblock 2009.
\newblock Uncertain agent verification through probabilistic model-checking.
\newblock In {\em Safety and Security in Multiagent Systems}.

\bibitem[\protect\citeauthoryear{Belardinelli \bgroup et al\mbox.\egroup
  }{2020}]{journals-BLMR20}
Belardinelli, F.; Lomuscio, A.; Murano, A.; and Rubin, S.
\newblock 2020.
\newblock Verification of multi-agent systems with public actions against
  strategy logic.
\newblock {\em Artif. Intell.} 285.

\bibitem[\protect\citeauthoryear{Berthon \bgroup et al\mbox.\egroup
  }{2020}]{berthon2020alternating}
Berthon, R.; Fijalkow, N.; Filiot, E.; Guha, S.; Maubert, B.; Murano, A.;
  Pinault, L.; Pinchinat, S.; Rubin, S.; and Serre, O.
\newblock 2020.
\newblock Alternating tree automata with qualitative semantics.
\newblock {\em ACM Trans. Comput. Logic} 22(1):1--24.

\bibitem[\protect\citeauthoryear{Berthon \bgroup et al\mbox.\egroup
  }{2021}]{tocl-BMMRV21}
Berthon, R.; Maubert, B.; Murano, A.; Rubin, S.; and Vardi, M.~Y.
\newblock 2021.
\newblock Strategy logic with imperfect information.
\newblock {\em {ACM} Trans. Comput. Log.} 22(1):1--51.

\bibitem[\protect\citeauthoryear{Berwanger and Doyen}{2008}]{BD08}
Berwanger, D., and Doyen, L.
\newblock 2008.
\newblock On the power of imperfect information.
\newblock In Hariharan, R.; Mukund, M.; and Vinay, V., eds., {\em Proc. of
  {FSTTCS} 2008}, volume~2 of {\em LIPIcs},  73--82.

\bibitem[\protect\citeauthoryear{Bulling and
  Jamroga}{2009}]{DBLP:journals/fuin/BullingJ09}
Bulling, N., and Jamroga, W.
\newblock 2009.
\newblock What agents can probably enforce.
\newblock {\em Fundam. Informaticae} 93(1-3):81--96.

\bibitem[\protect\citeauthoryear{Bulling and
  Jamroga}{2011}]{bulling2011alternating}
Bulling, N., and Jamroga, W.
\newblock 2011.
\newblock Alternating epistemic mu-calculus.
\newblock In {\em Proc. of IJCAI 2011},  109--114.
\newblock {IJCAI/AAAI}.

\bibitem[\protect\citeauthoryear{Bulling and Jamroga}{2014}]{BJ14}
Bulling, N., and Jamroga, W.
\newblock 2014.
\newblock Comparing variants of strategic ability: how uncertainty and memory
  influence general properties of games.
\newblock {\em Journal of Autonomous Agents and Multi-Agent Systems}
  28(3):474--518.

\bibitem[\protect\citeauthoryear{Carayol, L{\"o}ding, and
  Serre}{2018}]{carayol2018pure}
Carayol, A.; L{\"o}ding, C.; and Serre, O.
\newblock 2018.
\newblock Pure strategies in imperfect information stochastic games.
\newblock {\em Fundamenta Informaticae} 160(4):361--384.

\bibitem[\protect\citeauthoryear{Cerm{\'{a}}k \bgroup et al\mbox.\egroup
  }{2018}]{journals-CLMM18}
Cerm{\'{a}}k, P.; Lomuscio, A.; Mogavero, F.; and Murano, A.
\newblock 2018.
\newblock Practical verification of multi-agent systems against {SLK}
  specifications.
\newblock {\em Inf. Comput.} 261:588--614.

\bibitem[\protect\citeauthoryear{Chatterjee, Henzinger, and
  Piterman}{2010}]{ChatterjeeHP10}
Chatterjee, K.; Henzinger, T.~A.; and Piterman, N.
\newblock 2010.
\newblock {S}trategy {L}ogic.
\newblock {\em Inf. Comput.} 208(6):677--693.

\bibitem[\protect\citeauthoryear{Chen and Lu}{2007}]{chen2007probabilistic}
Chen, T., and Lu, J.
\newblock 2007.
\newblock Probabilistic alternating-time temporal logic and model checking
  algorithm.
\newblock In {\em Proc. of {FSKD}},  35--39.

\bibitem[\protect\citeauthoryear{Chen \bgroup et al\mbox.\egroup
  }{2011}]{chen2011verifying}
Chen, T.; Kwiatkowska, M.; Parker, D.; and Simaitis, A.
\newblock 2011.
\newblock Verifying team formation protocols with probabilistic model checking.
\newblock In {\em Proc. of {CLIMA} 2011}, LNCS 6814,  190--207.
\newblock Springer.

\bibitem[\protect\citeauthoryear{Clarke \bgroup et al\mbox.\egroup
  }{2018a}]{clarke2018handbook}
Clarke, E.~M.; Henzinger, T.~A.; Veith, H.; and Bloem, R.
\newblock 2018a.
\newblock {\em Handbook of Model Checking}.
\newblock Springer Publishing Company, Incorporated, 1st edition.

\bibitem[\protect\citeauthoryear{Clarke \bgroup et al\mbox.\egroup
  }{2018b}]{clarke2018model}
Clarke, E.; Grumberg, O.; Kroening, D.; Peled, D.; and Veith, H.
\newblock 2018b.
\newblock {\em Model checking}.
\newblock MIT press.

\bibitem[\protect\citeauthoryear{Dima and Tiplea}{2011}]{Dima11undecidable}
Dima, C., and Tiplea, F.
\newblock 2011.
\newblock Model-checking {ATL} under imperfect information and perfect recall
  semantics is undecidable.
\newblock {\em CoRR} abs/1102.4225.

\bibitem[\protect\citeauthoryear{Do \bgroup et al\mbox.\egroup
  }{2022}]{do2022online}
Do, V.; Hervouin, M.; Lang, J.; and Skowron, P.
\newblock 2022.
\newblock Online approval committee elections.
\newblock In {\em Proc. of {IJCAI 2022}},  251--257.
\newblock ijcai.org.

\bibitem[\protect\citeauthoryear{Doyen and Raskin}{2011}]{doyen2011games}
Doyen, L., and Raskin, J.-F.
\newblock 2011.
\newblock Games with imperfect information: theory and algorithms.
\newblock {\em Lectures in Game Theory for Computer Scientists} 10.

\bibitem[\protect\citeauthoryear{Doyen}{2022}]{Doyen2022}
Doyen, L.
\newblock 2022.
\newblock Stochastic games with synchronizing objectives.
\newblock In {\em Proc. of LICS}, LICS '22.
\newblock New York, NY, USA: Association for Computing Machinery.

\bibitem[\protect\citeauthoryear{Fagin \bgroup et al\mbox.\egroup
  }{2004}]{fagin2004reasoning}
Fagin, R.; Halpern, J.~Y.; Moses, Y.; and Vardi, M.
\newblock 2004.
\newblock {\em Reasoning about knowledge}.
\newblock MIT press.

\bibitem[\protect\citeauthoryear{Fu \bgroup et al\mbox.\egroup
  }{2018}]{fu2018model}
Fu, C.; Turrini, A.; Huang, X.; Song, L.; Feng, Y.; and Zhang, L.
\newblock 2018.
\newblock Model checking probabilistic epistemic logic for probabilistic
  multiagent systems.
\newblock In {\em Proc. of {IJCAI} 2018},  4757--4763.

\bibitem[\protect\citeauthoryear{Gripon and
  Serre}{2009}]{gripon2009qualitative}
Gripon, V., and Serre, O.
\newblock 2009.
\newblock Qualitative concurrent stochastic games with imperfect information.
\newblock In {\em Proc. of ICALP 2009},  200--211.
\newblock Springer.

\bibitem[\protect\citeauthoryear{Gurov, Goranko, and
  Lundberg}{2022}]{gurov2022knowledge}
Gurov, D.; Goranko, V.; and Lundberg, E.
\newblock 2022.
\newblock Knowledge-based strategies for multi-agent teams playing against
  nature.
\newblock {\em Artificial Intelligence} 309:103728.

\bibitem[\protect\citeauthoryear{Gutierrez \bgroup et al\mbox.\egroup
  }{2021}]{KR2021-30}
Gutierrez, J.; Hammond, L.; Lin, A.~W.; Najib, M.; and Wooldridge, M.
\newblock 2021.
\newblock {Rational Verification for Probabilistic Systems}.
\newblock In {\em Proc. of {KR} 2021},  312--322.

\bibitem[\protect\citeauthoryear{Hansson and Jonsson}{1994}]{Hansson94logic}
Hansson, H., and Jonsson, B.
\newblock 1994.
\newblock A logic for reasoning about time and reliability.
\newblock {\em Formal Aspects of Computing} 6(5):512--535.

\bibitem[\protect\citeauthoryear{Hao \bgroup et al\mbox.\egroup
  }{2012}]{hao2012probabilistic}
Hao, J.; Song, S.; Liu, Y.; Sun, J.; Gui, L.; Dong, J.~S.; and Leung, H.-f.
\newblock 2012.
\newblock Probabilistic model checking multi-agent behaviors in dispersion
  games using counter abstraction.
\newblock In {\em Proc. of {PRIMA} 2012}, LNCS 7455,  16--30.
\newblock Springer.

\bibitem[\protect\citeauthoryear{Huang and Luo}{2013}]{huang2013logic}
Huang, X., and Luo, C.
\newblock 2013.
\newblock A logic of probabilistic knowledge and strategy.
\newblock In {\em Proc. of {AAMAS} 2013},  845--852.

\bibitem[\protect\citeauthoryear{Huang, Su, and Zhang}{2012}]{Huang2012}
Huang, X.; Su, K.; and Zhang, C.
\newblock 2012.
\newblock Probabilistic alternating-time temporal logic of incomplete
  information and synchronous perfect recall.
\newblock In {\em Proc. of {AAAI} 2012},  765--771.

\bibitem[\protect\citeauthoryear{Jamroga and {\AA}gotnes}{2007}]{JA06}
Jamroga, W., and {\AA}gotnes, T.
\newblock 2007.
\newblock Constructive knowledge: what agents can achieve under imperfect
  information.
\newblock {\em J. Applied Non-Classical Logics} 17(4):423--475.

\bibitem[\protect\citeauthoryear{Jamroga and
  Bulling}{2011}]{jamroga2011comparing}
Jamroga, W., and Bulling, N.
\newblock 2011.
\newblock Comparing variants of strategic ability.
\newblock In {\em Proc. of {IJCAI} 2011},  252--257.
\newblock {IJCAI/AAAI}.

\bibitem[\protect\citeauthoryear{Jamroga and Dix}{2006}]{Jamroga06atlir-eumas}
Jamroga, W., and Dix, J.
\newblock 2006.
\newblock Model checking abilities under incomplete information is indeed
  delta2-complete.
\newblock In {\em Proc. of {EUMAS} 2006}, {CEUR} 223.
\newblock CEUR-WS.org.

\bibitem[\protect\citeauthoryear{Kemeny, Snell, and
  Knapp}{1976}]{kemeny1976stochastic}
Kemeny, J.~G.; Snell, J.~L.; and Knapp, A.~W.
\newblock 1976.
\newblock Stochastic processes.
\newblock In {\em Denumerable Markov Chains}. Springer.
\newblock  40--57.

\bibitem[\protect\citeauthoryear{Kupferman and Vardi}{2000}]{KV00}
Kupferman, O., and Vardi, M.~Y.
\newblock 2000.
\newblock Synthesis with incomplete informatio.
\newblock In {\em Advances in Temporal Logic}. Berlin: Springer.
\newblock  109--127.

\bibitem[\protect\citeauthoryear{Kwiatkowska \bgroup et al\mbox.\egroup
  }{2022}]{Kwiatkowska2022}
Kwiatkowska, M.; Norman, G.; Parker, D.; Santos, G.; and Yan, R.
\newblock 2022.
\newblock Probabilistic model checking for strategic equilibria-based decision
  making: Advances and challenges.
\newblock In {\em Proc. of MFCS 2022},  4--22.

\bibitem[\protect\citeauthoryear{Laroussinie and
  Mar\-key}{2015}]{DBLP:journals/iandc/LaroussinieM15}
Laroussinie, F., and Mar\-key, N.
\newblock 2015.
\newblock Augmenting {ATL} with strategy contexts.
\newblock {\em Inf. Comput.} 245:98--123.

\bibitem[\protect\citeauthoryear{Lomuscio and
  Pirovano}{2020}]{lomuscio2020parameterised}
Lomuscio, A., and Pirovano, E.
\newblock 2020.
\newblock Parameterised verification of strategic properties in probabilistic
  multi-agent systems.
\newblock In {\em Proc. of {AAMAS} 2020},  762--770.

\bibitem[\protect\citeauthoryear{Maubert \bgroup et al\mbox.\egroup
  }{2021}]{SLKF_KR21}
Maubert, B.; Mittelmann, M.; Murano, A.; and Perrussel, L.
\newblock 2021.
\newblock Strategic reasoning in automated mechanism design.
\newblock In {\em KR-21}.

\bibitem[\protect\citeauthoryear{Mittelmann \bgroup et al\mbox.\egroup
  }{2022}]{DBLP:conf/ijcai/MittelmannMMP22}
Mittelmann, M.; Maubert, B.; Murano, A.; and Perrussel, L.
\newblock 2022.
\newblock Automated synthesis of mechanisms.
\newblock In {\em {IJCAI-22}}.

\bibitem[\protect\citeauthoryear{Mittelmann \bgroup et al\mbox.\egroup
  }{2023}]{mittelmann2023bayesian}
Mittelmann, M.; Maubert, B.; Murano, A.; and Perrussel, L.
\newblock 2023.
\newblock Formal verification of bayesian mechanisms.
\newblock In {\em Proc. of {AAAI} 2023}.

\bibitem[\protect\citeauthoryear{Mogavero \bgroup et al\mbox.\egroup
  }{2014}]{MMPV14}
Mogavero, F.; Murano, A.; Perelli, G.; and Vardi, M.~Y.
\newblock 2014.
\newblock Reasoning about strategies: On the model-checking problem.
\newblock {\em {ACM} Trans. Comput. Log.} 15(4).

\bibitem[\protect\citeauthoryear{Nguyen and
  Rakib}{2019}]{nguyen2019probabilistic}
Nguyen, H.~N., and Rakib, A.
\newblock 2019.
\newblock A probabilistic logic for resource-bounded multi-agent systems.
\newblock In {\em Proc. of {IJCAI} 2019},  521--527.

\bibitem[\protect\citeauthoryear{Nisan \bgroup et al\mbox.\egroup
  }{2007}]{Nisan2007}
Nisan, N.; Roughgarden, T.; Tardos, {\'{E}}.; and Vazirani, V.
\newblock 2007.
\newblock {\em Algorithmic Game Theory}.
\newblock Cambridge Univ. Press.

\bibitem[\protect\citeauthoryear{Reif}{1984}]{Reif84}
Reif, J.~H.
\newblock 1984.
\newblock The complexity of two-player games of incomplete information.
\newblock {\em Journal of Computer and System Sciences} 29(2):274--301.

\bibitem[\protect\citeauthoryear{Schobbens}{2004}]{DBLP:journals/entcs/Schobbens04}
Schobbens, P.
\newblock 2004.
\newblock Alternating-time logic with imperfect recall.
\newblock {\em Electr. Notes Theor. Comput. Sci.} 85(2):82--93.

\bibitem[\protect\citeauthoryear{Song \bgroup et al\mbox.\egroup
  }{2019}]{song2019probabilistic}
Song, F.; Zhang, Y.; Chen, T.; Tang, Y.; and Xu, Z.
\newblock 2019.
\newblock Probabilistic alternating-time $\mu$-calculus.
\newblock In {\em Proc. of {AAAI} 2019},  6179--6186.

\bibitem[\protect\citeauthoryear{Vlassis, Littman, and
  Barber}{2012}]{Vlassis12memorylessPOMDP}
Vlassis, N.; Littman, M.~L.; and Barber, D.
\newblock 2012.
\newblock On the computational complexity of stochastic controller optimization
  in {POMDPs}.
\newblock {\em {ACM} Trans. Comput. Theory} 4(4):12:1--12:8.

\bibitem[\protect\citeauthoryear{Wan, Bentahar, and Hamza}{2013}]{wan2013model}
Wan, W.; Bentahar, J.; and Hamza, A.~B.
\newblock 2013.
\newblock Model checking epistemic--probabilistic logic using probabilistic
  interpreted systems.
\newblock {\em Knowledge-Based Systems} 50:279--295.

\end{thebibliography}

\end{document}